\documentclass[twocolumn,final,twoside,journal]{IEEEtran}

\usepackage{amsmath,amsthm,graphicx,cite}
\usepackage{srcltx}
\usepackage{epsfig,amsfonts,subfigure}
\usepackage{graphicx,cite,amssymb,amsmath}
\usepackage{color}
\usepackage[usenames,dvipsnames]{xcolor}
\usepackage{tabu}

\usepackage[nolist]{acronym}
\usepackage{psfrag}
\usepackage{perso}
\usepackage{pgfplots}
 \pgfplotsset{compat=newest}
    \pgfplotsset{plot coordinates/math parser=false}
    \pgfplotsset{
    label style={anchor=near ticklabel},
    xlabel style={yshift=0.0em},
    ylabel style={yshift=-0.3em},
    tick label style={font=\footnotesize },
    label style={font=\footnotesize},
    legend style={font=\footnotesize},
    title style={font=\fontsize{7}}}
\newtheorem{mydef}{Definition}
\newtheorem{prop}{Proposition}
\newtheorem{theorem}{Theorem}

\usepackage{xcolor}
\definecolor{iso}{rgb}{0.7,0.7,0.7}
\usepackage{blindtext}
\usepackage{flushend}
\usepackage{relsize}

\usepgflibrary{arrows}
\usetikzlibrary{patterns}
\usepgflibrary{decorations.pathmorphing}

\newcommand{\figw}{0.99\columnwidth}

\newcommand{\giangio}{}

\newcommand{\vecu}{\mathbf{u}}
\newcommand{\vecv}{\mathbf{v}}
\newcommand{\vecx}{\mathbf{x}}
\newcommand{\vecU}{\mathbf{U}}
\newcommand{\vecV}{\mathbf{V}}
\newcommand{\vecX}{\mathbf{X}}

\newcommand{\veca}{\mathbf{a}}

\newcommand{\hw}{w_{\mathsf H}}
\renewcommand{\deg}{\mathrm{deg}}

\newcommand{\field}{\mathbb{F}_2}

\newcommand{\ro}{r_{\mathsf o}}
\newcommand{\ri}{r_{\mathsf i}}

\newcommand{\ensemble}{\msr{C}}
\newcommand{\oensemble}{\msr{C}_{\mathsf o}}

\newcommand{\dmax}{d_{\mathrm{max}}}
\newcommand{\dmin}{\bar d_{\text{min}}}

\newcommand{\we}{A}
\newcommand{\weo}{A^{\mathsf o}}
\newcommand{\wei}{A^{\mathsf i}}

\newcommand{\tw}{\tilde{w}}
\newcommand{\tl}{\tilde{l}}
\newcommand{\Hb}{\mathsf H_{\mathsf b}}

\newcommand{\fmax}{\mathsf f_{\textrm{\textnormal{max}}}}
\newcommand{\f}{\mathsf f}

\newcommand{\inner}{\msr{I}}
\renewcommand{\outer}{\msr{O}}

\newcommand{\region}{\msr{P}}

\begin{document}
\begin{acronym}
\acro{WE}{weight enumerator}
\acro{WEF}{weight enumerator function}
\acro{IOWEF}{input output weight enumerator function}
\acro{IOWE}{input output weight enumerator}
\acro{LT}{Luby Transform}
\acro{BP}{belief propagation}
\acro{ML}{maximum likelihood}
\acro{MDS}{maximum distance separable}
\acro{LDPC}{low density parity check}
\acro{i.i.d.}{independent and identically distributed}
\end{acronym}


\title{On The Weight Distribution\\ of Fixed-Rate Raptor Codes}

\author{
    \IEEEauthorblockN{Francisco L\'azaro\IEEEauthorrefmark{1}, Enrico Paolini\IEEEauthorrefmark{2}, Gianluigi Liva\IEEEauthorrefmark{1}, Gerhard Bauch\IEEEauthorrefmark{3}}\\
    \IEEEauthorblockA{\IEEEauthorrefmark{1}Institute of Communications and Navigation of DLR (German Aerospace Center),
    \\Wessling, Germany. Email: \{Francisco.LazaroBlasco,Gianluigi.Liva\}@dlr.de}\\
    \IEEEauthorblockA{\IEEEauthorrefmark{2}Department of Electrical, Electronic, and Information Engineering, University of Bologna,
    \\Cesena, Italy. Email: e.paolini@unibo.it}\\
    \IEEEauthorblockA{\IEEEauthorrefmark{3}Institute for Telecommunication, Hamburg University of Technology
    \\Hamburg, Germany. Email: Bauch@tuhh.de}
    \thanks{This work will be presented at the 2015 IEEE International Symposium on Information Theory (ISIT), Hong Kong , China}

\thanks{\copyright 2015 IEEE. Personal use of this material is permitted. Permission
from IEEE must be obtained for all other uses, in any current or future media, including
reprinting /republishing this material for advertising or promotional purposes, creating new
collective works, for resale or redistribution to servers or lists, or reuse of any copyrighted
component of this work in other works}
}

\maketitle



\thispagestyle{empty} \pagestyle{empty}

\begin{abstract}
In this paper Raptor code ensembles with linear random precodes in a fixed-rate setting are considered. An expression for the average distance spectrum is derived and
this expression is used to obtain the asymptotic exponent of the weight distribution.
The asymptotic growth rate analysis is then exploited to develop a necessary and sufficient condition under which the fixed-rate Raptor code ensemble exhibits a strictly positive typical minimum distance.
\end{abstract}


\section{Introduction}\label{sec:Intro}

Fountain codes \cite{byers02:fountain} are erasure codes potentially able to generate an endless amount of encoded symbols.  As such, they find application in contexts where the channel erasure rate is not a priori known.
 The first class of practical fountain codes, \ac{LT} codes, was introduced in \cite{luby02:LT} together with an iterative \ac{BP} decoding algorithm that is efficient when the number of input symbols $k$ is large. One of the shortcomings of \ac{LT} codes is that in order to have a low probability of unsuccessful decoding, the encoding cost per output symbol has to be $\mathcal O \left(\ln(k)\right)$.
Raptor codes \cite{shokrollahi06:raptor} overcome this problem. They consist of a serial concatenation of an outer precode $\mathcal C$ with an inner \ac{LT} code. The \ac{LT} code design can thus be relaxed requiring only the recovery of a fraction $1-\gamma$ of the input symbols with $\gamma$ small. This can be achieved with linear encoding complexity. The outer precode is responsible for recovering the remaining fraction of input symbols, $\gamma$. If the precode $\mathcal C$ is linear-time encodable, then the Raptor code has a linear encoding complexity, $\mathcal O\left( k \right)$, and, therefore, the overall encoding cost per output symbol is constant with respect to $k$. Furthermore, Raptor codes are universally capacity-achieving on the binary erasure channel.

Most of the works on \ac{LT} and Raptor codes consider \ac{BP} decoding which has a good performance for very large input blocks ($k$ at least in the order of a few tens of thousands symbols). Often, in practice smaller blocks are used. For example, for the Raptor codes standardized in \cite{MBMS12:raptor} and \cite{luby2007rfc} the recommended values of $k$ range from $1024$ to $8192$. For these input block lengths, the performance under \ac{BP} decoding degrades considerably. In this context, an efficient \ac{ML} decoding algorithm in the form of inactivation decoding \cite{shokrollahi2005systems} may be used in place of \ac{BP}. Recently,   \ac{ML} decoding for Raptor and \ac{LT} codes has been analyzed \cite{Lazaro:ITW104,Lazaro:SCC15,mahdaviani2012raptor}, focusing however mainly on their decoding complexity under inactivation decoding.

Despite their rateless capability, Raptor codes represent an excellent solution for fixed-rate communication schemes requiring powerful erasure correction capabilities with low decoding complexity. Hence, it is not surprising that Raptor codes are actually used in a fixed-rate setting by existing communication systems (see e.g. \cite{DVB-SH:raptor}). In this context, the performance under  \ac{ML} erasure decoding is determined by the distance properties of the fixed-rate Raptor code ensemble, that to the best knowledge of the authors have not yet been analyzed.

In this paper we analyze the distance properties of fixed-rate Raptor codes. In particular, we focus on the case where the precode is picked from the linear random ensemble. The choice of this ensemble is not arbitrary. The precode used by some standardized Raptor codes \cite{MBMS12:raptor,luby2007rfc} is a concatenation of two systematic codes, the first being a high-rate regular \ac{LDPC} code and the second being pseudo-random code characterized a dense parity check matrix. These precodes were designed to behave like codes of the linear random ensemble in terms of rank properties, but allowing a fast algorithm for matrix vector multiplication \cite{ShokrollahiNow:2009}. Thus, we conjecture that the results obtained for the ensemble considered in this work may give (as a first approximation) hints on the distance properties of Raptor codes employed in existing systems.
For this Raptor code ensemble we develop a necessary and sufficient condition to guarantee a strictly positive typical minimum distance. The condition is found to depend on the degree distribution of the inner \ac{LT} code and on the code rates of both the inner \ac{LT} code and the (outer) precode. A necessary condition is also derived which, beyond the inner/outer code rates, depends on the average output degree only.

The rest of the paper is organized as follows. The main definitions are introduced in Section~\ref{sec:ensemble}.
Section~\ref{sec:dist} provides the derivation of the average weight distribution of this Raptor code ensemble and of the associated growth rate.
Section~\ref{sec:rate_reg} provides necessary and sufficient conditions for a positive typical minimum distance. The conclusions follow in Section~\ref{sec:Conclusions}.


\section{Preliminaries}\label{sec:ensemble}

We consider fixed-rate Raptor code ensembles based on the encoder structure depicted in Figure \ref{fig:raptor}. The encoder is given by a serial concatenation of an $(h,k)$ outer precode with an $(n,h)$ inner fixed-rate \ac{LT} code. We denote by $\vecu$ the outer encoder input, and by $\vecU$ the corresponding random vector. Similarly, $\vecv$ and $\vecx$ denote the input and the output of the fixed-rate \ac{LT} encoder, with $\vecV$ and $\vecX$ being the corresponding random vectors. The vectors $\vecu$, $\vecv$ and $\vecx$ are composed by $k$, $h$ and $n$ symbols each. The symbols of $\vecu$ are referred to as \emph{source} symbols, whereas the symbols of $\vecv$ and $\vecx$ are referred to as \emph{intermediate} and \emph{output} symbols, respectively.

We restrict to symbols belonging to $\field$. We denote by $\hw(\veca)$ the Hamming weight of a binary vector $\veca$. For a generic \ac{LT} code output symbol $x_i$, $\deg (x_i)$ denotes the output symbol degree, i.e., the number of intermediate symbols that are added (in $\mathbb F_2$) to produce $x_i$. We will respectively denote by $\ro=k/h$,  $\ri=h/n$, and $r=k/n=\ro \ri$ the rates of the outer code, the inner \ac{LT} code, and the Raptor code.

We consider the ensemble of Raptor codes $\ensemble(\oensemble,\Omega, \ri, \ro, n)$ obtained by a serial concatenation of an outer code in the $\left(\ri n,\ro\ri n\right)$ binary linear random block code ensemble, $\oensemble$, with all possible realizations of an $\left(n,\ri n\right)$ fixed-rate \ac{LT} code with output degree distribution $\Omega= \{ \Omega_1, \Omega_2,\Omega_3, \ldots, \Omega_{\dmax}\}$, where $\Omega_i$ corresponds to the probability of having an output symbol of degree $i$. We finally denote as $\bar \Omega$ the average output degree, $
\bar \Omega = \sum_i i\Omega_i$.

In the following we make use of the notion of exponential equivalence \cite{CoverThomasBook}, as follows. Two real-valued positive  sequences $a(n)$ and $b(n)$ are said to be exponentially equivalent, writing $a(n)~\doteq~b(n)$, when
\begin{equation}\label{eq:asymp_eq}
\lim_{n \to \infty} \frac{1}{n} \log_2 \frac{a(n)}{b(n)}=0.
\end{equation}
Moreover, given two pairs of reals $(x_1,y_1)$ and $(x_2,y_2)$, we write $(x_1,y_1) \succeq (x_2,y_2)$ when $x_1 \geq x_2$ and $y_1 \geq y_2$.

\begin{figure}
        \centering
        {\includegraphics[width=\figw]{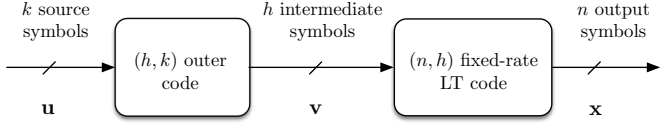}}
        \caption{Raptor codes consist of a serial concatenation of a linear block code (pre-code) with a \ac{LT} code.}
        \label{fig:raptor}
\end{figure}

\section{Distance Spectrum of Fixed-Rate Raptor Code Ensembles}\label{sec:dist}
\subsection{Average Weight Enumerator}

Let us denote by $A_w$ the average \ac{WE} of the ensemble $\ensemble(\oensemble,\Omega, \ri, \ro, n)$. For $w >0$ we have
\begin{equation}
A_w = \sum_{l=1}^{h} \frac{\weo_l \wei_{l,w}}{ \binom {h} {l}}
\label{eq:we_serial}
\end{equation}
where $\weo_l$ is the average \ac{WE} of the outer precode, and $\wei_{l,w}$ is the average \ac{IOWE} of the inner \ac{LT} code. The average \ac{WE} of an $(h,k)$ random code is known to be (see \cite{Gallager63})
\begin{equation}
\weo_l = \binom{h}{l} 2^{-h (1-\ro)}.
\label{eq:wef_random}
\end{equation}

We now focus on the average \ac{IOWE} of the \ac{LT} code.  Let us denote by $l$ the Hamming weight of the input of the \ac{LT} encoder, and let us assume that the output symbol of the \ac{LT} code has degree $j$. Let us denote by $p_{j,l}$ the probability that any of the $n$ output bits of the \ac{LT} encoder takes the value $1$ given that the Hamming weight of the intermediate word  is $l$ and the degree of the \ac{LT} code output symbol is $j$, i.e.,
\[
p_{j,l}:=\Pr\{X_i=1|\hw(\vecV)=l,\deg(X_i)=j\}
\]
for any $i\in \{1,\dots,n\}$. This probability may be expressed as
\begin{equation}
p_{j,l} =
\sum_{\substack{i=\max(1,l+j-h)\\ i~\textrm{odd}}} ^{ \min (l,j)} \frac{ \binom{j}{i} \binom{h-j}{l-i} } { \binom{h}{l} }.
\label{eq:p_j_l}
\end{equation}
Removing the conditioning on $j$ we obtain $p_l$, the probability of any of the $n$ output bits of the \ac{LT} encoder taking value $1$ given a Hamming weight $l$ for the intermediate word, i.e.,
\[
p_{l}:=\Pr\{X_i=1|\hw(\vecV)=l\}
\]
for any $i\in \{1,\dots,n\}$.
We have
\begin{equation*}
p_l = \sum_{j=1}^{\dmax} \Omega_j p_{j,l}.\label{eq:pl_finite}
\end{equation*}
Since the output bits are generated by the \ac{LT} encoder independently of each other, the Hamming weight of the \ac{LT} codeword conditioned to an intermediate word of weight $l$ is a binomially distributed random variable with parameters $n$ and $p_l$. Hence, we have
\begin{equation}
\Pr\{\hw(\vecX) = w | \hw(\vecV) = l\}
=\binom {n}{w} p_l^w (1-p_l)^{n-w}.\label{eq:distr_weight_LT}
\end{equation}
The average \ac{IOWE} of a \ac{LT} code may be now easily calculated multiplying \eqref{eq:distr_weight_LT} by the number of weight-$l$ intermediate words, yielding
\begin{equation}
\wei_{l,w}= \binom {h}{l} \binom {n}{w} p_l^w (1-p_l)^{n-w}.
\label{eq:iowef_lt}
\end{equation}
Making use of \eqref{eq:we_serial}, \eqref{eq:wef_random} and \eqref{eq:iowef_lt}, for $w>0$ the average \ac{WE}   of the fixed-rate Raptor code ensemble can be expressed as
\begin{equation}
A_w = \binom {n}{w} 2^{-h (1-\ro)} \sum_{l=1}^h \binom{h}{l}   p_l^w (1-p_l)^{n-w}.\label{eq:WEF_Raptor}
\end{equation}


\subsection{Growth Rate of Fixed-Rate Raptor Code Ensembles}\label{sec:growth}

In this subsection we compute the asymptotic exponent (growth rate) of the weight distribution for the ensemble $\msr{C}_{\infty}(\oensemble,\Omega, \ri, \ro)$, that is the ensemble $\msr{C}(\oensemble,\Omega, \ri, \ro, n)$ in the limit where $n$ tends to infinity for constant $\ri$ and $\ro$.
Hereafter, we denote the normalized output weight of the Raptor encoder by $\tilde w = w/n$ and the normalized output weight of the precoder (input weight to the \ac{LT} encoder) by $\tilde l = l/h$.

\begin{figure}
\begin{center}
\includegraphics[width=1.03\columnwidth]{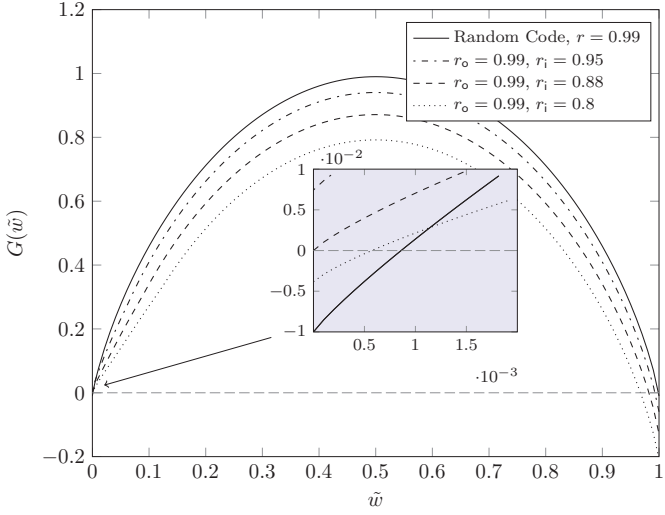}
\centering \caption{Growth rate vs. normalized output weight $\tw$. The solid line shows the growth rate of a linear random code with rate $r=0.99$. The dot-dashed, dashed, and dotted lines show the growth rates $G(\tw)$ of the ensemble $\msr{C}_{\infty}(\oensemble,\Omega^{(2)}, \ri, \ro=0.99)$ for $\ri=0.95$, $0.88$ and $0.8$, respectively.}
\label{fig:growth}
\end{center}
\end{figure}

Using the well-known exponential equivalence ${n \choose n \tw}~\doteq~2^{n \Hb(\tw)}$, where $\Hb$ is the binary entropy function, for large $n$  the multiplicative term in front of the summation in \eqref{eq:WEF_Raptor} fulfills
\begin{equation*}
\binom {n}{w} 2^{-h (1-\ro)} \doteq 2^{n[\Hb(\tw) -\ri(1-\ro)]} := \beta \, .
\label{eq:wef_random_asymp}
\end{equation*}
Therefore, $A_{w}=A_{\tw n}$ in \eqref{eq:WEF_Raptor} fulfills
\begin{align*}
A_{ \tw n} \doteq \beta \sum_{\tl} 2^ {n \left[ \ri \Hb(\tl) + \tw \log_2 p_{\tl} + (1- \tw) \log_2 \left(1 - p_{\tl} \right) \right]}
\end{align*}
where
\begin{align*}
p_{\tl} = \sum_{j=1}^{\dmax} \Omega_j p_{j,\tl}
\end{align*}
and
\begin{align}
p_{j,\tl} = \frac{1}{2} \left[  1-\left( 1-2\tilde l\right)^j \right].
\nonumber
\end{align}
Using the result
\begin{align*}
\sum_\alpha 2^{n f(\alpha)} \doteq \max_\alpha 2^{n f(\alpha)}
\end{align*}
we can simplify the expression of $A_{\tw n}$ as
\begin{align*}
A_{\tw n}  \doteq 2^{n \left[ \Hb(\tw) - \ri  (1-\ro)+ \fmax(\tilde w)\right]}
\end{align*}
where
\[
\fmax(\tilde w) := \max_{\tl} \f(\tw, \tl)
\]
and
\begin{align*}
\f(\tw, \tl) := \ri \Hb(\tl) + \tw \log_2 p_{\tl} + (1- \tw) \log_2 \left(1 - p_{\tl}\right).
\end{align*}
The asymptotic exponent of the fixed-rate Raptor code ensemble weight distribution is finally
\begin{align*}
G(\tilde w) &:= \lim_{n \to \infty} \frac{1}{n} \log_2 A_{\tw n}\nonumber\\
&\phantom{:}= \Hb(\tw) - \ri  (1-\ro) +  \fmax(\tw).
\label{eq:growth_rate}
\end{align*}
Moreover, the real number
\begin{align*}
\bar d_{\text{min}} := \inf \{ \tw>0 : G(\tw) > 0 \}
\end{align*}
is the typical minimum distance of the ensemble.

Fig.~\ref{fig:growth} shows $G(\tw)$ for the ensemble $\msr{C}_{\infty}(\oensemble,\Omega^{(2)}, \ri, \ro)$, where $\Omega^{(2)}$ is the output degree distribution used in the standards \cite{MBMS12:raptor}, \cite{luby2007rfc} (see details in Table~\ref{table:dist}) and $\ro=0.99$ for three different  $\ri$ values. It can be observed how the curve for $\ri = 0.95$ does not  cross the $x$-axis, the curve for $\ri = 0.88$ has $\bar d_{\text{min}}=0$ and the curve for $\ri=0.8$ has $\bar d_{\text{min}}=0.0005$.
The figure also shows the growth rate of the precode, a linear random code with $r=0.99$. It can be observed how the typical minimum distance of the precode is larger than that of the concatenated (Raptor) code.


Fig.~\ref{fig:gilbert} shows the overall rate $r$ of the Raptor code ensemble $\msr{C}_{\infty}(\oensemble,\Omega^{(2)}, \ri=r/ \ro, \ro)$ versus the typical minimum distance $\dmin$. It can be observed how, for constant overall rate $r$, $\dmin$ increases as the outer code rate $\ro$ decreases. It also can be observed how decreasing $\ro$ allows to get closer to the asymptotic Gilbert-Varshamov bound.

\begin{figure}
\begin{center}
\includegraphics[width=1.01\columnwidth]{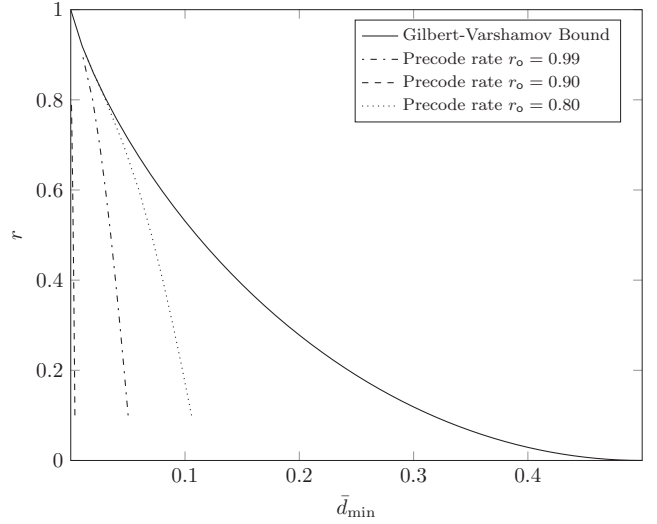}
\centering \caption{Overall rate $r$ vs. the typical minimum distance $\dmin$. The solid line represents the asymptotic Gilbert-Varshamov bound. The different dashed and dotted lines represent Raptor codes ensembles $\msr{C}_{\infty}(\oensemble,\Omega^{(2)}, \ri=r/ \ro, \ro)$ with different outer code rates, $\ro$.}
\label{fig:gilbert}
\end{center}
\end{figure}


\section{Rate Regions}
\label{sec:rate_reg}
We are now interested in determining whether the ensemble exhibits good typical distance properties. More specifically, we are interested in the existence of a strictly positive typical minimum distance. A sufficient condition for having a positive typical minimum distance is
\begin{equation*}
\lim_{\tw \to 0^+} G(\tw) < 0
\end{equation*}
which implies
\begin{equation}
\ri  (1-\ro) >  \lim_{\tw \to 0^+}\fmax(\tw).
\label{eq:lim_G}
\end{equation}
Unfortunately, a closed-form expression for the right-hand side of \eqref{eq:lim_G} does not exist in general. However, the order of limit and the maximization can be inverted by observing that the function $\fmax$ is right-continuous at $\tw=0$, that is
\begin{align*}
\lim_{\tw \rightarrow 0^+} \fmax(\tw) =  \fmax(0).
\end{align*}
It is now possible to recast the right-hand side of \eqref{eq:lim_G} as
\begin{align}
\lim_{\tw \to 0^+}\fmax(\tw)&=\lim_{\tw \to 0^+} \max_{\tl} \f(\tw, \tl) =  \max_{\tl} \lim_{\tw \to 0^+} \f(\tw, \tl) \nonumber \\
&= \max_{\tl} \left[ \ri \Hb(\tl) + \log_2 \left(1 - p_{\tl}\right)\right] \label{eq:G_0} \nonumber\\
&=: \fmax^*(\ri )
\end{align}
where we emphasized that $\fmax^*$ hides a dependency on $\ri$.
Computing \eqref{eq:G_0} implies carrying out a maximization which cannot generally be computed analytically. However, the function to be maximized is sufficient \emph{well behaved} so that the maximization can be done numerically in an efficient manner.

\begin{mydef}[Positive typical minimum distance region] We define the \emph{positive} typical minimum distance region of a Raptor code ensemble as the set $\region$ of code rate pairs $\left( \ri, \ro \right)$ for which the ensemble possesses a positive typical minimum distance. Formally :
\begin{align}
\region:=\left\{(\ri,\ro) \succeq (0,0) | \bar d_{\text{min}} (\Omega, \ri,\ro)> 0 \right\},
\nonumber
\end{align}
Where we have used the notation $ \bar d_{\text{min}}=  \bar d_{\text{min}}(\Omega, \ri,\ro)$ to emphasize the dependence on $\Omega$, $\ri$ and $\ro$.
\end{mydef}


\begin{theorem} An inner positive typical minimum distance region, $\inner$, is given by
\label{theorem_inner}
\begin{equation}
\inner:=\left\{ (\ri,\ro) \succeq (0,0) | \ri  (1-\ro) > \fmax^*(\ri)\right\}.
\nonumber
\end{equation}
\end{theorem}
\begin{proof}
It follows from \eqref{eq:lim_G} being a sufficient condition for having a positive typical minimum distance.
\end{proof}

\begin{theorem}
\label{theorem_necessary}
The inner positive typical minimum distance region $\inner$ and the positive typical minimum distance $\region$ region coincide, $\inner \equiv \region$.
\end{theorem}
\begin{proof}
Due to space constraints we provide only a sketch of the proof. The argument is based on the observation that any pair $(\ri,\ro)$ such that $\ri  (1-\ro) <  \lim_{\tw \to 0^+}\fmax(\tw)$ cannot belong to $\region$ since for these pairs $\lim_{\tw\rightarrow 0^+} G(\tw)>0$. An analysis must then be carried out for those $(\ri,\ro)$ pairs such that $\ri  (1-\ro) =  \lim_{\tw \to 0^+}\fmax(\tw)$, meaning $\lim_{\tw\rightarrow 0^+} G(\tw)=0$. For these $(\ri,\ro)$ pairs, the only possibility for having a positive typical minimum distance is
\begin{equation}
\lim_{\tw \rightarrow 0^+} G'(\tw) < 0. \nonumber
\end{equation}
The proof is completed by showing that the above condition never holds, regardless of $\lim_{\tw \to 0^+} G(\tw)$ (that is regardless the considered $(\ri,\ro)$ pair).
\end{proof}

Although Theorems \ref{theorem_inner} and \ref{theorem_necessary} fully characterize the positive typical minimum distance $\region$, they require the calculation of
$\fmax^*$. In the following we introduce an outer region that can be computed more easily, and only depends on the average output degree.
\begin{prop}
The positive typical minimum distance region $region$ of a fixed-rate Raptor code ensemble $\msr{C}_{\infty}(\oensemble,\Omega, \ri, \ro)$ fulfills $\region \subseteq \outer$, where
\begin{equation}
\outer := \left\{(\ri,\ro) \succeq (0,0) | \ri \leq \min \left( \phi(\ro), \frac{1}{\giangio{\ro}}\right) \right\}
 \label{eq:outer_bound_2}
\end{equation}
with
\begin{equation}
\phi(\ro)= \frac{\bar \Omega \log_2 \ro}{\Hb(1-\ro) -(1-\ro)}. \nonumber
\end{equation}
\end{prop}
\begin{proof}
The proof goes by lower bounding $\we_{\tw n}$ for $\tw \to 0^+$. Observing \eqref{eq:we_serial} we can see how $\we_{\tw n}$ is summation of the number of Hamming weight $\tw n$ codewords generated by all possible input weights to the LT encoder. A lower bound to $\we_{\tw n}$
is the number of Hamming weight $\tw n$ codewords generated only by inputs to the LT encoder of weight $\tl~=~1-\ro$. Manipulating the expression obtained and making use of the exponential equivalency introduced in Section~\ref{sec:ensemble} the expression in \eqref{eq:outer_bound_2} is obtained.
%
\end{proof}

It is important to point out that the asymptotic exponent of the weight distribution captures linear-sized codewords \cite{di06:weight}. Codewords whose weight grows with a sublinear weight should be subject to an ad-hoc analysis.

\begin{table}[t]
\caption{Degree distributions $\Omega^{(1)}$, defined in \cite{MBMS12:raptor,luby2007rfc}  and $\Omega^{(2)}$, defined in \cite{shokrollahi06:raptor}}
\begin{center}
\begin{tabular}{|c|c|c|c|c|c|c|c|c|c|c|c|c|c|c|c|}
\hline
  Degree& $\Omega^{(1)}$ & $\Omega^{(2)}$ \rule{0pt}{2.6ex} \rule[-0.9ex]{0pt}{0pt} \\ \hline\hline
  ${1}$ & 0.0098 & 0.0048 \\ \hline
  ${2}$ & 0.4590 & 0.4965 \\ \hline
  ${3}$ & 0.2110 & 0.1669 \\ \hline
  ${4}$ & 0.1134 & 0.0734 \\ \hline
  ${5}$ &  &  0.0822  \\ \hline
  ${8}$ &  & 0.0575 \\ \hline
  ${9}$ &  & 0.0360 \\ \hline
  ${10}$ & 0.1113 &  \\ \hline
  ${11}$ & 0.0799 &  \\ \hline
  ${18}$ &  & 0.0012 \\ \hline
  ${19}$ &  & 0.0543  \\ \hline
  ${40}$ & 0.0156 &  \\ \hline
  ${65}$ &  &  0.0182 \\ \hline
  ${66}$ &  & 0.0091 \\ \hline \hline
  $\bar \Omega$ & 4.6314 & 5.825  \rule{0pt}{2.6ex} \rule[-0.9ex]{0pt}{0pt} \\
    \hline
\end{tabular}
\end{center}\label{table:dist}
\end{table}

We now consider two different output degree distributions given in Table.~\ref{table:dist}. The first one is the output degree distribution used in the standards \cite{MBMS12:raptor}, \cite{luby2007rfc}, which we will refer to as $\Omega^{(1)}$. Then, we
consider a distribution $\Omega^{(2)}$ which was designed in \cite{shokrollahi06:raptor} for $k= 120000$.

In  Fig.~\ref{fig:region} we show the positive typical minimum distance region, $\region$ for $\Omega^{(1)}$ and $\Omega^{(2)}$ together with their outer bound to the positive growth region $\outer$. It can be observed how the outer bound is tight in both cases except for inner codes rates close to $\ri=1$. The figure also shows several isorate curves, along which the overall rate of the Raptor code $r$ stays constant. For example, in order to have a positive typical minimum distance and an overall rate $r=0.95$, the figure shows that the rate of the precode must lay below $\ro<0.978$ for both distributions. It is quite remarkable that for precode rates below $\ro<0.978$ there exist a region in which $\Omega^{(1)}$ exhibits a positive typical minimum distance and $\Omega^{(2)}$ does not, although the average output degree of $\Omega^{(1)}$ is considerably lower than that of $\Omega^{(2)}$. This exemplifies the fact that the distance properties of a Raptor code ensemble depend not only on $\ri$ and $\ro$ and $\bar \Omega$ but also on the degree distribution $\Omega$.

\begin{figure}[t]
        \centering
        \includegraphics[width=1.03\columnwidth]{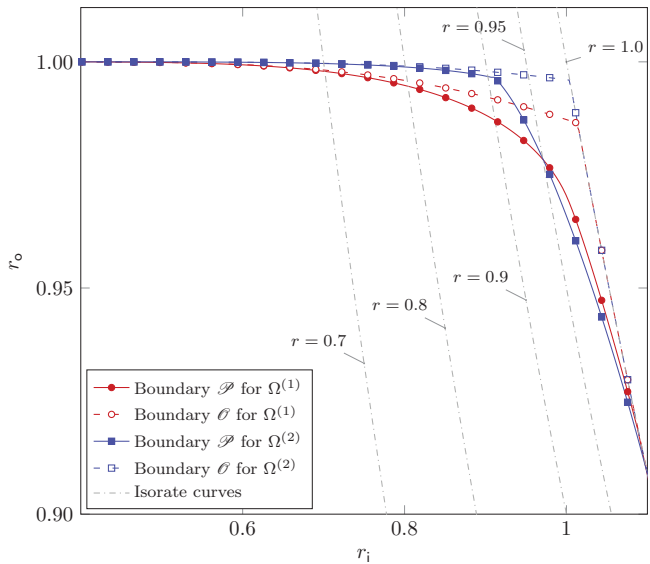}
        \caption{Positive growth rate region. The solid and dotted lines represent the positive growth-rate region of $\Omega^{(1)}$ and $\Omega^{\mathrm{(2)}}$ and the dashed line represents its outer bound. The markers represent the rate point at which codes in the standards \cite{MBMS12:raptor}, \cite{luby2007rfc} operate for different values of $k$.}
\label{fig:region}
\end{figure}


\section{Conclusions}\label{sec:Conclusions}
In this work we have considered ensembles of fixed-rate Raptor codes which use linear random codes as precodes. We have derived the expression of the average \acl{WE} of an ensemble and the expression of the growth rate of the \acl{WE} as functions of the rate of the outer code and the rate and degree distribution of the inner \ac{LT} code.
Based on these expressions we are able to determine necessary and sufficient conditions to have Raptor code ensembles with a positive typical minimum distance. A simple necessary condition has been developed too, which requires  (besides the inner and outer code rates) the knowledge of the average output degree only. Despite the fact that only binary Raptor codes have been considered, an extension to higher order fields is possible with a limited effort.

The work presented in this paper helps to understand the behavior of fixed-rate Raptor codes under \ac{ML} decoding and it can be used to design Raptor codes with good distance properties, for example, using numerical optimization.

\section{Acknowledgements}\label{sec:Acknowledgements}
The authors would like to acknowledge Dr. Massimo Cicognani and Dr. Mark Flanagan for the useful discussions. 

\bibliographystyle{IEEEtran}
\bibliography{IEEEabrv,Raptor}

\begin{thebibliography}{10}
\providecommand{\url}[1]{#1}
\csname url@samestyle\endcsname
\providecommand{\newblock}{\relax}
\providecommand{\bibinfo}[2]{#2}
\providecommand{\BIBentrySTDinterwordspacing}{\spaceskip=0pt\relax}
\providecommand{\BIBentryALTinterwordstretchfactor}{4}
\providecommand{\BIBentryALTinterwordspacing}{\spaceskip=\fontdimen2\font plus
\BIBentryALTinterwordstretchfactor\fontdimen3\font minus
  \fontdimen4\font\relax}
\providecommand{\BIBforeignlanguage}[2]{{%
\expandafter\ifx\csname l@#1\endcsname\relax
\typeout{** WARNING: IEEEtran.bst: No hyphenation pattern has been}%
\typeout{** loaded for the language `#1'. Using the pattern for}%
\typeout{** the default language instead.}%
\else
\language=\csname l@#1\endcsname
\fi
#2}}
\providecommand{\BIBdecl}{\relax}
\BIBdecl

\bibitem{byers02:fountain}
J.~Byers, M.~Luby, and M.~Mitzenmacher, ``A digital fountain approach to
  reliable distribution of bulk data,'' \emph{{IEEE} J. Select. Areas Commun.},
  vol.~20, no.~8, pp. 1528--1540, Oct. 2002.

\bibitem{luby02:LT}
M.~Luby, ``{LT} codes,'' in \emph{Proc. of the 43rd Annual IEEE Symp. on
  Foundations of Computer Science}, Vancouver, Canada, Nov. 2002, pp. 271--282.

\bibitem{shokrollahi06:raptor}
M.~Shokrollahi, ``Raptor codes,'' \emph{{IEEE} Trans. Inf. Theory}, vol.~52,
  no.~6, pp. 2551--2567, Jun. 2006.

\bibitem{MBMS12:raptor}
{3GPP TS 26.346 V11.1.0}, ``Technical specification group services and system
  aspects; multimedia broadcast/multicast service; protocols and codecs,'' Jun.
  2012.

\bibitem{luby2007rfc}
M.~Luby, A.~Shokrollahi, M.~Watson, and T.~Stockhammer, ``{RFC} 5053: Raptor
  forward error correction scheme: Scheme for object delivery,'' {IETF}, Tech.
  Rep., Oct. 2007.

\bibitem{shokrollahi2005systems}
M.~Shokrollahi, S.~Lassen, and R.~Karp, ``Systems and processes for decoding
  chain reaction codes through inactivation,'' Feb. 2005, {US} Patent
  6,856,263.

\bibitem{Lazaro:ITW104}
F.~L\'azaro~Blasco, G.~Liva, and G.~Bauch, ``{LT} code design for inactivation
  decoding,'' in \emph{Proc. 2014 IEEE Inf. Theory Workshop}, Hobart,
  Australia, Nov. 2014.

\bibitem{Lazaro:SCC15}
------, ``Enhancing the {LT} component of {Raptor} codes for inactivation
  decoding,'' in \emph{Proc. 2015 International {ITG} Conf. on Systems, Commun.
  and Coding}, Hamburg, Germany, Feb. 2015.

\bibitem{mahdaviani2012raptor}
K.~Mahdaviani, M.~Ardakani, and C.~Tellambura, ``{On Raptor code design for
  inactivation decoding},'' \emph{{IEEE} Commun. Lett.}, vol.~60, no.~9, pp.
  2377--2381, Sep. 2012.

\bibitem{DVB-SH:raptor}
{ETSI TR 102 993 V1.1.1}, ``{Digital Video Broadcasting (DVB)}; upper layer
  {FEC} for {DVB} systems,'' Feb. 2011.

\bibitem{ShokrollahiNow:2009}
A.~Shokrollahi and M.~Luby, ``Raptor codes,'' \emph{Found. and Trends on
  Commun. and Inf. Theory}, vol.~6, pp. 213--322, Mar. 2009.

\bibitem{CoverThomasBook}
T.~M. Cover and J.~A. Thomas, \emph{Elements of information theory},
  2nd~ed.\hskip 1em plus 0.5em minus 0.4em\relax New York: Wiley, 2006, chapter
  3.

\bibitem{Gallager63}
R.~G. Gallager, ``Low-density parity-check codes,'' Ph.D. dissertation, Dep.
  Electrical Eng., M.I.T, Cambridge, MA, Jul. 1963.

\bibitem{di06:weight}
C.~Di, R.~Urbanke, and T.~Richardson, ``Weight distribution of low-density
  parity-check codes,'' \emph{{IEEE} Trans. Inf. Theory}, vol.~52, no.~11, pp.
  4839--4855, Nov. 2006.

\end{thebibliography}


\end{document}